\documentclass[11pt,hidelinks]{article}

\usepackage[sc]{mathpazo}
\usepackage[T1]{fontenc}
\usepackage[margin=1in]{geometry}
\usepackage[table]{xcolor}
\usepackage{proof}
\usepackage[font=small]{caption}
\usepackage{bbm}
\usepackage{microtype}
\usepackage{xspace}
\usepackage{amsmath}
\usepackage{amssymb}
\usepackage{mathtools}
\usepackage{enumitem}
\usepackage{array}

\newcolumntype{L}[1]{>{\raggedright\let\newline\\\arraybackslash\hspace{0pt}}m{#1}}

\usepackage{authblk}

\usepackage[sort,nocompress]{cite}

\usepackage[colorlinks,citecolor=blue]{hyperref}
\usepackage{amsthm}
\newtheorem{theorem}{Theorem}[section]
\newtheorem{lemma}[theorem]{Lemma}
\newtheorem{definition}[theorem]{Definition}

\newtheorem{corollary}[theorem]{Corollary}

\newcommand{\eq}[1]{\hyperref[eq:#1]{(\ref*{eq:#1})}}
\renewcommand{\sec}[1]{\hyperref[sec:#1]{Section~\ref*{sec:#1}}}
\newcommand{\thm}[1]{\hyperref[thm:#1]{Theorem~\ref*{thm:#1}}}
\newcommand{\lem}[1]{\hyperref[lem:#1]{Lemma~\ref*{lem:#1}}}
\newcommand{\cor}[1]{\hyperref[cor:#1]{Corollary~\ref*{cor:#1}}}
\newcommand{\app}[1]{\hyperref[app:#1]{Appendix~\ref*{app:#1}}}
\newcommand{\dfn}[1]{\hyperref[dfn:#1]{Definition~\ref*{dfn:#1}}}
\newcommand{\fig}[1]{\hyperref[fig:#1]{Figure~\ref*{fig:#1}}}
\newcommand{\tab}[1]{\hyperref[tab:#1]{Table~\ref*{tab:#1}}}

\newcommand{\bra}[1]{\langle #1 \vert}
\newcommand{\ket}[1]{\vert #1 \rangle}
\newcommand{\proj}[1]{\vert #1\rangle\!\langle #1\vert}

\newcommand{\tr}[0]{\mathrm{tr}}
\newcommand{\Tr}[0]{\mathrm{Tr}}
\newcommand{\Span}[0]{\mathrm{span}}

\DeclareMathOperator*{\Exp}{\mathbb{E}}

\newcommand{\1}[0]{\mathbbm{1}}
\newcommand{\A}[0]{\mathcal{A}}

\renewcommand{\d}[0]{\mathrm{d}}
\newcommand{\F}[0]{\mathbb{F}}
\newcommand{\Fq}[0]{\mathbb{F}_q}
\newcommand{\R}[0]{\mathbb{R}}
\newcommand{\mv}[0]{\mathsf{Mv}}
\newcommand{\vm}[0]{\mathsf{vM}}

\newcommand{\vmv}[0]{\mathsf{vMv}}

\newcommand{\mvphase}[0]{\widetilde{\mathsf{Mv}}}
\newcommand{\swap}[0]{\mathsf{SWAP}}
\renewcommand{\Re}{\mathop{\mathrm{Re}}}

\begin{document}

\title{Quantum query complexity with matrix-vector products}
\author[1]{Andrew M.\ Childs}
\author[1]{Shih-Han Hung}
\author[1,2]{Tongyang Li}

\affil[1]{Joint Center for Quantum Information and Computer Science, Department of Computer Science, and Institute for Advanced Computer Studies, University of Maryland}

\affil[2]{Center for Theoretical Physics, Massachusetts Institute of Technology}
\date{}

\maketitle

\begin{abstract}
We study quantum algorithms that learn properties of a matrix using queries that return its action on an input vector. We show that for various problems, including computing the trace, determinant, or rank of a matrix or solving a linear system that it specifies, quantum computers do not provide an asymptotic speedup over classical computation. On the other hand, we show that for some problems, such as computing the parities of rows or columns or deciding if there are two identical rows or columns, quantum computers provide exponential speedup. We demonstrate this by showing equivalence between models that provide matrix-vector products, vector-matrix products, and vector-matrix-vector products, whereas the power of these models can vary significantly for classical computation.
\end{abstract}

\section{Introduction} \label{sec:intro}

Algorithms for linear algebra problems---for example, solving linear systems and determining basic properties of matrices such as rank, trace, determinant, eigenvalues, and eigenvectors---constitute a fundamental research area in applied mathematics and theoretical computer science. Such tasks have widespread applications in scientific computation, statistics, operations research, and many other related areas. Algorithmic linear algebra also provides a fundamental toolbox that can inspire the design of algorithms in general.

There are several possible models of access to a matrix, and linear-algebraic algorithms can depend significantly on how the input is represented (as discussed further below). One natural model is the \emph{matrix-vector product} ($\mv$) oracle. For a matrix $M\in\F^{n\times m}$ in a given field $\F$, the $\mv$ oracle takes $x\in\F^m$ as input and outputs $Mx\in\F^n$. Matrix-vector products arise, for example, as the elementary step of the power method (and the related Lanczos method) for computing the largest eigenvector of a matrix. Matrix-vector products also commonly appear in streaming algorithms, especially in the technique of sketching (see the survey \cite{Wood14} for more information).

Recent work has studied the classical complexity of various basic problems in the $\mv$ model. Specifically, Sun, Woodruff, Yang, and Zhang \cite{SWYZ19} studied the complexities of various linear algebra, statistics, and graph problems using matrix-vector products, and Braverman, Hazan, Simchowitz, and Woodworth \cite{BHSW20} proved tight bounds on maximum eigenvalue computation and linear regression in this model. Rashtchian, Woodruff, and Zhu \cite{RWZ20} considered a generalization to the vector-matrix-vector product ($\vmv$) oracle, which returns $x^\top\! M y$ for given input vectors $x \in \F^n, y \in \F^m$, and studied the complexity of various linear algebra, statistics, and graph problems in this setting. \tab{main} includes a partial summary of these results.

Quantum computers can solve certain problems much faster than classical computers, so it is natural to study quantum query complexity with matrix-vector products.
Lee, Santha, and Zhang recently studied the quantum query complexity of graph problems with cut queries \cite{LSZ20}, which are closely related to matrix-vector queries.
For a weighted graph $G=(V,w)$ where $|V|=n$ and $w$ assigns a nonnegative integer weight to each edge, the input of a cut query is a subset $S\subseteq V$ and the output is $|w(S,V\setminus S)|$, the total weight of the edges between $S$ and $V\setminus S$. This can be viewed as a version of the $\vmv$ model over $\mathbb{Z}$, with the extra assumptions that $x \in \{0,1\}^n, y \in \{0,1\}^m$ are both boolean and $M$ is a symmetric matrix with nonnegative integer entries. Reference \cite{LSZ20} gives quantum algorithms for determining all connected components of $G$ with $O(\log^{6}n)$ quantum cut queries, and for outputting a spanning forest of $G$ with $O(\log^{8}n)$ quantum cut queries. Both problems require $\Omega(n/\log n)$ classical cut queries, so the quantum algorithms provide exponential speedups.

In other recent work on structured queries for graph problems, Montanaro and Shao studied the problem of learning an unknown graph with ``parity queries'' \cite{MS20}:
for an unknown graph with adjacency matrix $A$, the parity oracle takes as input a string $x$ that encodes a subset of the vertices, and returns $x^\top\! A x \bmod 2$. This query model is the $\vmv$ model over $\F_2$ with the extra restriction that the left and right vectors are identical.

Van Apeldoorn and Gribling studied Simon's problem for linear functions over a prime field $\F_p$ \cite{vAG18}. In this problem, the oracle encodes a linear function $f\colon\F_p\to\F_p$, and the task is to determine if the function is one-to-one, or if there is a one-dimensional subspace $H\subset\F_p$ such that for every $x,x'\in\F_p^n$, $f(x)=f(x')$ if and only if $x-x'\in H$.
Such a function can be represented by a square matrix over $\F_p$, and the problem is equivalent to determining whether that matrix is full rank or has nullity 1 using matrix-vector product queries.

Other past work has developed linear algebraic quantum algorithms using different input models. Quantum algorithms for high-dimensional linear algebra have been studied extensively since Harrow, Hassidim, and Lloyd introduced a method for generating a quantum state proportional to the solution of a large, sparse system of linear equations \cite{HHL09}. This algorithm assumes a quantum oracle that determines the locations and values of the nonzero entries of a matrix in any given row or column, and the ability to generate a quantum state that encodes the right-hand side of the linear system. Subsequent work has led to improved and generalized algorithms under similar assumptions.
However, it is challenging to find practical applications that achieve speedup over classical computation \cite{Chi09,Aar15}. Recent work by Apers and de Wolf \cite{AdW20} gives polynomial quantum speedup for producing an explicit classical description of the solution of a Laplacian linear system, assuming adjacency-list access to the underlying graph of the Laplacian.
Note also that for various problems including determinant estimation, rank testing, linear regression, etc., there is a large separation between the classical query complexities under $\mv$ and entrywise queries ($\tilde{\Theta}(n)$ \cite{SWYZ19} and $\Theta(n^{2})$,
respectively). These results show how the model of access to a matrix can significantly impact the complexity of solving linear-algebraic problems.  A better understanding of the quantum matrix-vector oracle could therefore provide a useful tool for the design of future quantum algorithms.

\paragraph{Contributions.}
We conduct a systematic study of quantum query complexity with a matrix-vector oracle for a matrix $M\in\Fq^{m\times n}$, where $\Fq$ is a given finite field. Using this model, we provide results on the quantum query complexities of linear algebra and statistics problems.

First, we prove that various linear algebra problems, including
\begin{itemize}[noitemsep]
\item computing the trace $\tr(M)$ of $M\in\Fq^{n\times n}$;
\item computing the determinant $\det(M)$ of $M\in\Fq^{n\times n}$;
\item solving the linear system $Ax=b$ for $A\in\Fq^{n\times n}$; and
\item testing whether $\mathrm{rank}(M)=n$ or $\mathrm{rank}(M)\leq n/2$ for a matrix $M\in\mathbb{F}_q^{m\times n}$;
\end{itemize}
require $\Omega(n)$ quantum queries to the $\mv$ oracle. Since $O(n)$ queries suffice to determine the entire matrix, even classically, these results show that no quantum speedup is possible. (As a side effect, we improve the $\Omega(n/\log n)$ classical lower bound for trace computation \cite{SWYZ19} to $\Omega(n)$.)

Our quantum lower bound for trace computation applies results of Copeland and Pommersheim \cite{CP18} by viewing the problem as a special case of \emph{coset identification}. Our lower bounds for other linear algebra problems are all proved by the polynomial method \cite{BBC+01,Aar02}. We show how to symmetrize the success probability to a univariate polynomial, and then give a lower bound on the polynomial degree using an observation of Koiran, Nesme, and Portier\ \cite{KNP07}.

On the other hand, we determine the matrix-vector quantum query complexity of several statistics problems, including
\begin{itemize}[noitemsep]
\item computing the row and column parities of $M\in\F_{2}^{m\times n}$;
\item deciding if there exist two identical columns in $M\in\F_{2}^{m\times n}$; and
\item deciding if there exist two identical rows in $M\in\F_{2}^{m\times n}$.
\end{itemize}
Specifically, we prove that their quantum query complexities with an $\mv$ oracle are $O(1)$, $O(\log n)$, and $O(\log m)$, respectively. Compared to the classical bounds using either the $\mv$ oracle \cite{SWYZ19} or the $\vmv$ oracle \cite{RWZ20}, our quantum algorithms achieve \emph{exponential} quantum speedups.

Technically, these results build upon our observation that the quantum query complexities in the $\mv$ model under left or right multiplication are \emph{identical} (\thm{left-right-eq}). In particular, one right $\mv$ query can be simulated using one left $\mv$ query, and vice versa. In contrast, classically there is a significant difference between matrix-vector ($\mv$) and vector-matrix ($\vm$) queries---for example, computing the parity of rows over $\F_{2}$ only takes $O(1)$ $\mv$ queries, but computing the parity of columns over $\F_{2}$ requires $\Theta(n)$ $\mv$ queries. In contrast, for both problems a quantum computer can achieve the smaller query complexity by switching to the easier side.

\begin{table}
    \centering
    \renewcommand{\arraystretch}{1.4}
    \rowcolors{2}{gray!10}{white}
    \footnotesize
  \begin{tabular}{|L{28mm}|L{42mm}L{42mm}L{35.5mm}|}
    \hline \rowcolor{gray!50}
    Problem & Classical $\mv$ & Classical $\vmv$ & Quantum (this paper) \\
    \hline
    Trace
      & $O(n),\Omega(n/\log n)$ for matrix with entries in $\{0,1,\ldots,n^3\}$ \& queries with entries in $\{0,1,\ldots,n^C\}$, $C\in\mathbb{N}$ \cite{SWYZ19};\newline $\Theta(n)$ over $\Fq$ (\thm{trace})
      & $O(n),\Omega(n/\log n)$ for matrix with entries in $\{0,1,\ldots,n^3\}$ \& queries with entries in $\{0,1,\ldots,n^C\}$, $C\in\mathbb{N}$ \cite{RWZ20};\newline $\Theta(n)$ over $\Fq$ (\thm{trace})
      & $\Theta(n)$ over $\Fq$\newline(\thm{trace}) \\
    Linear regression
      & $\Theta(n)$ over $\mathbb{R}$ \cite{BHSW20};\newline $\Theta(n)$ over $\Fq$ (\thm{regression})
      & $\Theta(n^2)$ over $\Fq$ (\cor{regression-vmv})
      & $\Theta(n)$ over $\Fq$\newline(\thm{regression}) \\
    Rank testing
      & $k+1$ to distinguish rank $\leq k$ from $k'>k$ over $\mathbb{R}$ \cite{SWYZ19};\newline $\Theta(n)$ over $\Fq$ (\thm{rank-testing})
      & $\Omega(k^2)$ to distinguish rank $k$ from $k+1$ over $\Fq$ \cite{RWZ20};\newline $\Omega(n^{2-O(\epsilon)})$ for non-adaptive $(1\pm\epsilon)$-approximation over $\mathbb{R}$ \cite{RWZ20}
      & $\Theta(\min\{m,n\})$ to distinguish rank $\min\{m,n\}$ from $\leq \frac{1}{2}\min\{m,n\}$ over $\Fq$\newline(\thm{rank-testing}) \\
    Two identical columns
      & $O(n/m)$, $m=\Omega(\log(n/\epsilon))$\newline over $\mathbb{F}_2$ \cite{SWYZ19}
      & $O(n\log n),\Omega(n)$ over $\mathbb{F}_2$ \cite{RWZ20}
      & $O(\log n)$ over $\mathbb{F}_2$\newline(\cor{identical-rows}) \\
    Two identical rows
      & $O(\log m)$ over $\mathbb{F}_2$ \cite{SWYZ19}
      & $O(n\log n),\Omega(n)$ over $\mathbb{F}_2$ \cite{RWZ20}
      & $O(\log m)$ over $\mathbb{F}_2$\newline(\cor{identical-rows}) \\
    Majority of columns
      & $\Omega(n/\log n)$ for binary matrices over $\mathbb{R}$ \cite{SWYZ19}
      & $\Theta(n^2)$ over $\mathbb{F}_2$ \cite{RWZ20}
      & $O(1)$ for binary matrices over $\mathbb{R}$ (\cor{majority}) \\
    Majority of rows
      & $O(1)$ for binary matrices\newline over $\mathbb{R}$ \cite{SWYZ19}
      & $\Theta(n^2)$ over $\mathbb{F}_2$ \cite{RWZ20}
      & $O(1)$ for binary matrices over $\mathbb{R}$ (\cor{majority}) \\
    Parity of columns
      & $\Theta(n)$ over $\mathbb{F}_2$ \cite{SWYZ19}
      & $\Theta(n)$ over $\F_2$ (\lem{vmv-parities})
      & $O(1)$ over $\mathbb{F}_2$\newline(\cor{row-parities}) \\
    Parity of rows
      & $O(1)$ over $\mathbb{F}_2$ \cite{SWYZ19}
      & $\Theta(m)$ over $\F_2$ (\lem{vmv-parities})
      & $O(1)$ over $\mathbb{F}_2$\newline(\cor{row-parities})  \\
    \hline
  \end{tabular}
  \caption{Comparison of classical and quantum query complexities with matrix-vector ($\mv$) and vector-matrix-vector ($\vmv$) product oracles for an $m\times n$ matrix.
    For trace and linear regression, $m=n$.
    Known query complexities over $\R$ and $\Fq$ are included for completeness; results over different fields are incomparable in general.
  }
  \label{tab:main}
\end{table}

Our results are summarized in \tab{main}, including some implications of our results for classical query complexity and a few additional results over $\mathbb{R}$. Note that there can be large gaps between the classical query complexities with $\mv$ and $\vmv$ queries, but they are the same in the quantum setting due to an equivalence between quantum $\mv$ and $\vmv$ queries (\thm{vmv-eq}), which follows along similar lines to the equivalence between $\mv$ and $\vm$ queries. The $\mv$--$\vmv$ equivalence is closely related to a similar equivalence shown in the work of Lee, Santha, and Zhang \cite{LSZ20}, as we discuss further in \sec{vmv}.

\paragraph{Open questions.}
Our paper leaves several natural open questions for future investigation:
\begin{itemize}
\item For linear algebra problems such as those we studied, can we also prove quantum query lower bounds for matrices over the real field $\R$? Our proofs rely on the polynomial method, and it is unclear how to adapt them to a setting with continuous input.

\item Can we prove a quantum lower bound for the task of minimizing a quadratic form $f(x)=\frac{1}{2}x^{\top}Ax+b^{\top}x$, where $A\in\R^{n\times n}$ and $b\in\R^{n}$? Note that $f$ is minimized at $x=-A^{-1}b$, and we can determine the vector $b$ and implement $\mv$ queries to the matrix $A$ using fast quantum gradient computation \cite{Jor05}, so this is closely related to the previous open question. Quadratic form minimization is a special case of optimizing a convex function $f\colon\R^{n}\to\R$ by quantum evaluation queries, where previous works \cite{GKNS21,CCLW20,AGGW20} left a quadratic gap between the best known quantum upper and lower bounds of $\tilde{O}(n)$ and $\Omega(\sqrt{n})$, respectively.

\item For the finite field case, can we identify other problems with quantum speedup over the classical matrix-vector oracle, or find advantage compared to other quantum oracles such as entrywise queries?
\end{itemize}

\paragraph{Organization.}
We review necessary background in \sec{prelim}. We prove the equivalence of quantum matrix-vector and vector-matrix-vector product oracles in \sec{equivalence}. In \sec{linear-algebra}, we prove tight quantum query complexity lower bounds on various linear algebra problems, including trace, determinant, linear systems, and rank.

\section{Preliminaries}\label{sec:prelim}

\subsection{The quantum query model}\label{sec:query-model}

Given a set $X$ and an abelian group $G$, let $f\colon X\to G$ be a function.
Access to $f$ is provided by a black-box unitary operation $U_f\colon\ket{x,y}\mapsto\ket{x,y+f(x)}$ for all $x\in X$ and $y\in G$.
We call an application of $U_f$ a (standard) query.

For a finite abelian group $G$, the Fourier transform over $G$ is
\begin{align}
  F_G \coloneqq \frac{1}{|G|^{1/2}} \sum_{x\in G}\sum_{y\in\hat{G}} \chi_y(x) \ket{y}\!\bra{x},
\end{align}
where $\hat{G}$ is a complete set of characters of $G$, and $\chi_y\colon G\to\mathbb{C}$ denotes the $y^{th}$ character of $G$.
Since $\hat{G}\cong G$, we label elements of $\hat{G}$ using elements of $G$.
Note that $\chi_y$ is a group homomorphism, i.e., $\chi_y(x+z)=\chi_y(x)\chi_y(z)$.
In addition, the characters satisfy the orthogonality condition
\begin{align}\label{eq:characters-orthogonal}
  \frac{1}{|G|} \sum_{z\in G} \chi_y(z)^* \chi_w(z) = \delta_{yw}.
\end{align}

A phase query is defined as a standard query conjugated by the Fourier transform acting on the output register. In other words, for $x\in X$ and $y\in G$, a phase query acts as
\begin{align}\nonumber
  \ket{x,y}
  &\xmapsto{\mathbb{1}\otimes F_G^\dag} \frac{1}{|G|^{1/2}}\sum_{z\in G}\chi_{y}(z)^*\ket{x,z} \\\nonumber
  &\xmapsto{U_f} \frac{1}{|G|^{1/2}}\sum_{z\in G}\chi_{y}(z)^*\ket{x,z+f(x)} \\\label{eq:phase}
  &\xmapsto{\mathbb{1}\otimes F_G} \frac{1}{|G|}\sum_{z\in G}\chi_{y}(z)^*\chi_w(z+f(x))\ket{x,w} = \chi_y(f(x))\ket{x,y}.
\end{align}
The equality in \eq{phase} follows from the orthogonality condition in \eq{characters-orthogonal}.
Since one can simulate a phase query using a single standard query and vice versa, the query complexities of any problem are equal with these two models.

Over a finite field $\mathbb{F}_q$ for prime power $q=p^r$, the Fourier transform over $\Fq$ is the unitary transformation $\ket{x}\mapsto q^{-1/2}\sum_{y\in\Fq}e(xy)\ket{y}$,
where the exponential function $e\colon\Fq\to\mathbb{C}$ is defined as $e(z) \coloneqq e^{2\pi i\Tr_{\F_q/\F_p}(z)/p}$ and the trace function $\Tr_{\Fq/\F_p}\colon\Fq\to\mathbb{F}_p$ is defined as $\Tr_{\F_q/\F_p}(z) \coloneqq z+z^p+z^{p^2}+\dots+z^{p^{r-1}}$.

Over the field of real numbers, the quantum Fourier transform is
\begin{align}
  F_{\R} := \int_\R\d y\int_\R \d x \, e^{2\pi i yx}\ket{y}\!\bra{x}.
\end{align}
The basis states $\{\ket{x}:x\in\R\}$ are normalized to the Dirac delta function, i.e., for $x,x'\in\R$, $\langle x'|x\rangle=\delta(x-x')$.
Here the Dirac delta function $\delta$ satisfies $\int_\R\d x' \, \delta(x-x')f(x')=f(x)$ for any function $f$.
Furthermore, we have $\int_\R\d y \, e^{2\pi i y(x-x')}=\delta(x-x')$.
By direct calculation using these facts, $F^\dag_\R F_\R=\int_\R\d x \, \ket{x}\!\bra{x}=\mathbb{1}$.

While we can formally consider a model of query complexity over $\mathbb R$ with arbitrary precision, its practical instantiation requires discrete approximation. We can achieve precision $\epsilon$ by approximating real numbers with $s = O(\log(1/\epsilon))$ bits, and can then replace the continuous Fourier transform with the discrete Fourier transform over $\mathbb{Z}_{2^s}$.
It is straightforward to show that a discretized phase query over $\mathbb{Z}_{2^s}$ can be implemented by Fourier transforming a standard query that maps discretized inputs to discretized function values.

\subsection{The coset identification problem}
\label{sec:coset}

Copeland and Pommersheim studied a kind of quantum query problem that they call the \emph{coset identification problem} \cite{CP18}.
They define this problem in a generalized query model where the black box does not necessarily perform a standard or phase query, although their definition includes those cases.
In the coset identification problem, we fix a finite group $G$ and a subgroup $H \le G$. The algorithm is given access to a unitary transformation $\pi(g)$, where $\pi$ is a representation of $G$ on vector space $V$.
When $\pi$ is given, the vector space $V$ is called the representation space (or simply, the representation) of $G$ \cite[Chapter~1]{Ser77}.
The goal is to determine which coset of $H$ the unknown element $g \in G$ belongs to.

\begin{definition}[Coset identification problem \cite{CP18}]
  A coset identification problem for a finite group $G$ and subgroup $H\leq G$ is a 3-tuple $(\pi,V,F)$ such that
  \begin{itemize}[nosep]
  \item $\pi$ is a unitary representation of $G$ in the complex vector space $V$, and
  \item $F$ is a function constant on left cosets of $H\leq G$ and distinct on distinct cosets, i.e., $F(g)=F(g')$ if and only if $g'=gh$ for some $h\in H$.
  \end{itemize}
  Given a black box that performs the unitary transformation $\pi(g)$, the goal is to compute $F(g)$.
\end{definition}

Copeland and Pommersheim show that the optimal success probability of a $t$-query algorithm for a coset identification problem can be calculated by taking, over all irreps $Y$ of $H$, the maximum of the fraction of the induced representation $Y^\uparrow$ of $G$ shared with $V^{\otimes t}$.
Furthermore, the optimal algorithm can be non-adaptive.
For a representation $V$, let $I(V)$ denote the set of irreducible characters of $G$ appearing in $V$.

\begin{theorem}[Optimal success probability of coset identification {\cite[Corollary 5.7]{CP18}}]\label{thm:opt-cip}
  The optimal success probability of any $t$-query quantum algorithm $\A$ for the coset identification problem $(\pi,V,F)$ for finite group $G$ and subgroup $H\leq G$, under uniformly random inputs in $G$, is
  \begin{align}\label{eq:cip-max-pr}
    \Pr[\A^{\pi(g)} = F(g)] =
    \max_Y \frac{\dim Y_{V^{\otimes t}}^\uparrow}{\dim Y^\uparrow},
  \end{align}
  where the probability is maximized over all irreducible representations $Y$ of $H$, $Y^\uparrow$ is the induced representation of $G$, and $A_B$ is the maximal subrepresentation of $A$ such that $I(A_B)\subseteq I(B)$ for representations $A,B$.
\end{theorem}

The \emph{oracle discrimination problem} is the special case of the coset identification problem where $H$ is the trivial group, i.e., the function $F$ is injective. In this case, $Y^\uparrow=\Span\{\ket{g}:g\in G\}$.

\begin{corollary}[Optimal success probability of oracle discrimination {\cite[Theorem 4.2]{CP18}}]
  The optimal success probability of the oracle discrimination problem is
  \begin{align}
    \frac{1}{|G|}\sum_{i\in I(V^{\otimes t})} d_i^2,
  \end{align}
  where $I(V^{\otimes t})$ is the irrep content of $(\pi^{\otimes t},V^{\otimes t})$ and $d_i$ is the dimension of irrep $i\in I(V^{\otimes t})$.
\end{corollary}

We consider the complexity of standard queries in the matrix-vector model. In this model, oracle access to a matrix $M\in\F^{m\times n}$ for field $\F$ and positive integers $m,n$ is the unitary operation $U(M)\colon\ket{x,y}\mapsto\ket{x,y+Mx}$.
The map $U$ is a representation of the additive group of matrices since it is a group homomorphism satisfying $U(M)U(N)=U(M+N)$ for all matrices $M,N$ of the same dimensions.
A phase query is also a unitary representation since it is a standard query conjugated by a fixed unitary matrix (the quantum Fourier transform).

\subsection{The polynomial method}

We will use the polynomial method to obtain quantum lower bounds. Here we state a version for non-boolean functions as used in \cite{Aar02}.

\begin{lemma}\label{lem:polynomial-2t}
  Let $\A$ be a $t$-query quantum algorithm given access to the input $x\in [m]^n$ for $m,n\in\mathbb{Z}$ through oracle $U_x\colon \ket{i,j}\mapsto\ket{i,j+x_i}$ for $i \in [n]$ and $j \in [m]$.
  The acceptance probability of $\A$ on input $x$ is a degree-$(2t)$ polynomial in $x_1,\ldots,x_n$.
\end{lemma}

\section{Equivalence of matrix-vector and vector-matrix-vector products}\label{sec:equivalence}

In this section, we show that the matrix-vector and vector-matrix-vector models are equivalent, i.e., for any problem, the quantum query complexities in these models differ by at most a constant factor.
Furthermore, we show that in the matrix-vector model, left matrix-vector products and right matrix-vector products are equivalent.
This is in stark contrast to the classical case where these query complexities can differ significantly, as mentioned in \sec{intro} and discussed further below.

\subsection{Left and right matrix-vector queries}

We first show that left matrix-vector products and right matrix-vector products are equivalent.
\begin{theorem}\label{thm:left-right-eq}
Quantum query complexities in the left and right matrix-vector models over a finite field are identical. In particular, one right $\mv$ query can be simulated using one left $\mv$ query, and vice versa.
\end{theorem}

\begin{proof}
For input matrix $M\in\Fq^{n\times m}$, a matrix-vector ($\mv$) query applies the unitary transformation
\begin{align}
  U^{\mv}(M)\colon \ket{x,y}\mapsto\ket{x,y+Mx}
\end{align}
for every $x\in\Fq^m$ and $y\in\Fq^n$.
Conjugating by a quantum Fourier transform on the output register yields a phase query
\begin{align}\nonumber
  \ket{x,y}
  &\xmapsto{\mathbb{1}\otimes F_{\F_q^n}^\dag} q^{-1/2}\sum_z e(-y^\top\! z) \ket{x,z} \\\nonumber
  &\xmapsto{U^{\mv}(M)} q^{-1/2}\sum_z e(-y^\top\! z) \ket{x,z+Mx} \\\nonumber
  &\xmapsto{\mathbb{1}\otimes F_{\F_q^n}} q^{-1}\sum_{z,w} e(-y^\top\! z+w^\top(z+Mx)) \ket{x,w} \\\nonumber
  &= \sum_{w} \delta[y=w]e(-y^\top\! z+w^\top(z+Mx)) \ket{x,w} \\
  &= e(y^\top\! Mx) \ket{x,y}. \label{eq:mvphase}
\end{align}
We denote this unitary transformation by $U^{\mvphase}(M)$.

Conjugating a phase query by a swap gate, we have
\begin{align}\nonumber
  \ket{x,y}
  &\xmapsto{\swap} \ket{y,x} \\\nonumber
  &\xmapsto{U^{\mvphase}(M)} e(x^\top\! M y)\ket{y,x} \\\nonumber
  &\xmapsto{\swap} e(x^\top\! M y)\ket{x,y} \\\label{eq:swap}
  &= e(y^\top\! M^\top\! x)\ket{x,y}.
\end{align}
This yields $U^{\mvphase}(M^\top)$, which in turn gives $U^{\mv}(M^\top)$ upon conjugation by an inverse quantum Fourier transform on the output register.
Thus one can simulate the oracle $U^{\mv}(M^\top)$ using one query to $U^{\mv}(M)$, showing equivalence of the two models.
\end{proof}

In contrast to \thm{left-right-eq}, Sun, Woodruff, Yang, and Zhang show that for the task of computing the row parities of an $m\times n$ matrix $M$ over $\mathbb{F}_2$, the left query complexity is $\Omega(m)$, whereas the right query complexity is 1 \cite{SWYZ19}.
Thus we have shown that computing column parities over $\mathbb{F}_2$ in the $\mv$ model has quantum query complexity 1, significantly less than the classical query complexity of $\Omega(n)$.

\begin{corollary}\label{cor:row-parities}
  The query complexity of computing the row parities and the column parities of an $m\times n$ matrix over $\mathbb{F}_2$ is $1$.
\end{corollary}

Note that it is easy to understand the randomized query complexities of these problems in the $\vmv$ model.

\begin{lemma}\label{lem:vmv-parities}
The randomized query complexities of computing the row parities and the column parities of an $m\times n$ matrix over $\mathbb{F}_2$ are $\Theta(m)$ and $\Theta(n)$, respectively.
\end{lemma}

\begin{proof}
  Each query reveals one bit of information, while the row parities convey $m$ bits, giving a lower bound of $\Omega(m)$.
  An algorithm querying $(e_1,1^n),\ldots, (e_m,1^n)$ learns the row parities with probability 1, giving an upper bound of $m$.
  The query complexity of column parities follows immediately from the symmetry of the $\vmv$ oracle.
\end{proof}

The randomized query complexities of determining if there exist identical columns or identical rows are $\Theta(n/m)$ and $\Theta(\log m)$, respectively \cite{SWYZ19}.
\thm{left-right-eq} implies that for identical columns, there is an exponential quantum speedup.

\begin{corollary}\label{cor:identical-rows}
  The query complexities of deciding if there exist two identical columns and rows in a $m\times n$ matrix over $\F_2$ are $O(\log n)$ and $O(\log m)$, respectively.
\end{corollary}

\begin{proof}
  By \thm{left-right-eq}, it suffices to give an algorithm for determining if there are two identical rows.
  To make the proof self-contained, we describe the algorithm of Sun, Woodruff, Yang, and Zhang \cite[Section~4.2]{SWYZ19}.
  The algorithm makes $q$ random queries $v_1,\ldots,v_q$, the entries of which are sampled uniformly.
  The algorithm outputs 1 if and only if there exist two entries $i,j$ such that $(Mv_k)_i=(Mv_k)_j$ for $k\in[q]$.

  To analyze the performance, for any two identical rows $m_i^\top,m_j^\top$, $\Pr_v[m_i^\top\! v=m_j^\top\! v]=1$. For $m_i\neq m_j$, $\Pr_v[m_i^\top\! v=m_j^\top\! v]\leq 1/2$.
  Therefore for a matrix that has two identical rows, the algorithm outputs 1 with probability 1.
  On the other hand, for a matrix that has no identical rows, the algorithm outputs 1 with probability
  \begin{align}\nonumber
    \Pr_{v_1,\ldots,v_q}[\exists i,j\in[m],~\forall \ell\in[q], m_i^\top\! v_\ell=m_j^\top\! v_\ell]
    &\leq
    \sum_{i,j\in[m],i\neq j}\Pr_{v_1,\ldots,v_q}[\forall \ell\in[q], m_i^\top\! v_\ell=m_j^\top\! v_\ell] \\
    &\leq
    \binom{m}{2} 2^{-q}.
  \end{align}
  Taking $q=2\log m$, the probability is no more than $\frac{1}{2}-\frac{1}{2m}$.
\end{proof}

The equivalence of left and right queries also holds over the reals.

\begin{theorem}\label{thm:left-right-eq-reals}
  Quantum query complexities in the left and the right matrix-vector models over $\R$ are identical.
  In particular, one right $\mv$ query can be simulated using one left $\mv$ query, and vice versa.
\end{theorem}

\begin{proof}
  The same idea as in the proof of \thm{left-right-eq} applies. 
  First, a phase query can be simulated by conjugating a standard query by the quantum Fourier transform.
  This yields $U^{\widetilde{\mv}}(M)$.
  Conjugating a phase query by a swap gate gives $U^{\widetilde{\mv}}(M^\top)$ with the same calculation as in \eq{swap}.
  This in turn yields $U^{\mv}(M^\top)$ upon conjugating $U^{\widetilde{\mv}}(M^\top)$ by an inverse quantum Fourier transform.
\end{proof}

Note that with finite precision, a phase query can be simulated using the quantum Fourier transform over an integer modulus (see \sec{query-model} for details).

As an example, we determine the query complexity of the majority of rows or columns: given a binary matrix $M\in\{0,1\}^{m\times n}$, compute the majority of each row or column over $\mathbb{R}$.

\begin{corollary}\label{cor:majority}
  The query complexities of computing the majorities of rows and columns of an $m \times n$ matrix over $\mathbb R$ are $1$.
\end{corollary}

\begin{proof}
  By \thm{left-right-eq-reals}, it suffices to show the query complexity of the majority of rows is 1.
  With a single query $(1,1,\ldots,1)^\top$, the majority of each row is determined.
\end{proof}

This result is not significantly affected by considering computation with finite precision. The number of 1s in each row and each column is an integer in $[0,k]$ for $k=\max\{m,n\}$. Thus a truncation with $O(\log k)$ bits suffices to perform the computation with no error.

\subsection{The vector-matrix-vector model}
\label{sec:vmv}

We now relate the power of the matrix-vector and vector-matrix-vector query models.
In the vector-matrix-vector model, the algorithm is given access to $M$ via $U^{\vmv}\colon\ket{x,y,a}\mapsto\ket{x,y,a+y^\top\! M x}$.
We can simulate one $\vmv$ query using two $\mv$ queries and an ancilla space storing a matrix-vector product:
\begin{align}\nonumber
  \ket{x,y,a}
  &\xmapsto{U^{\mv}(M)} \ket{x,y,a}\ket{Mx} \\\nonumber
  &\xmapsto{} \ket{x,y,a+y^\top\! Mx}\ket{Mx} \\
  &\xmapsto{U^{\mv}(M)^\dag} \ket{x,y,a+y^\top\! Mx}\ket{0}.
\end{align}

On the other hand, an $\mv$ phase query (defined previously in \eq{mvphase}) can be simulated using a $\vmv$ phase query by setting $a=1$:
\begin{align}
  \ket{x,y,1} \xmapsto{} e(y^\top\! M x)\ket{x,y,1}.
\end{align}
Such a $\vmv$ phase query can be constructed using one application of $U^{\vmv}$:
\begin{align}\nonumber
  \ket{x,y,a}
  &\xmapsto{\mathbb{1}\otimes\mathbb{1}\otimes F_{\F_q}^\dag} \sum_be(-ab)\ket{x,y,b} \\\nonumber
  &\xmapsto{U^{\vmv}(M)} \sum_be(-ab)\ket{x,y,b+y^\top\! Mx} \\\nonumber
  &\xmapsto{\mathbb{1}\otimes\mathbb{1}\otimes F_{\F_q}} \sum_{bc} e(-ab+c(b+y^\top\! Mx))\ket{x,y,c} \\
  &= e(ay^\top\! Mx)\ket{x,y,a}.
\end{align}

Thus we have shown the following.

\begin{theorem}\label{thm:vmv-eq}
Quantum query complexities in the matrix-vector and vector-matrix-vector models differ by at most a constant factor.
In particular, one $\vmv$ query can be simulated using two $\mv$ queries, and one $\mv$ query can be simulated using one $\vmv$ query.
\end{theorem}

This is again in stark contrast to the classical case, where the $\mv$ model can be much more powerful than the $\vmv$ model.
For example, for distinguishing a full-rank matrix from a rank-$(n-1)$ matrix, the randomized query complexity in the $\vmv$ model is $\Omega(n^2)$ \cite{RWZ20}, while the randomized query complexity in the $\mv$ model is $O(n)$ \cite{SWYZ19}.

Note that Lee, Santha, and Zhang \cite{LSZ20} previously studied the equivalence between quantum $\mv$ and $\vmv$ oracles. They focus on the special case where the matrix $M$ is the adjacency matrix of a graph with nonnegative integer weights and the inputs $x \in \{0,1\}^n, y \in \{0,1\}^m$ are boolean. In that setting, they prove equivalence between the $\vmv$ oracle and the additive oracle $a\colon 2^{[n]}\to\mathbb{Z}$ that returns $a(S)=\sum_{(u,v)\in S^{(2)}}w(u,v)$ for $S\subseteq [n]$, where $S^{(2)}$ denotes the set of cardinality-$2$ subsets of $S$. They also study relationships with other oracles that encode specific information about graphs (cuts, disjoint cuts, etc.; see Section 4 of \cite{LSZ20}). In contrast, our \thm{left-right-eq}, \thm{left-right-eq-reals}, and \thm{vmv-eq} work for inputs and matrices in fields, and do not apply to other graph oracles. While these results are, strictly speaking, incomparable, they are closely related,  both following from a generalization of the Bernstein-Vazirani algorithm \cite{BV97}.

\section{Linear algebra over finite fields}\label{sec:linear-algebra}

We now consider the quantum query complexity of particular linear algebra problems in the matrix-vector query model. Specifically, we consider learning the trace (\sec{trace}), computing the null space and determinant (\sec{nullspace}), solving linear systems (\sec{regression}), and estimating the rank (\sec{rank}).

\subsection{Trace} \label{sec:trace}

In this section, we show that the quantum query complexity of computing the trace of an $n\times n$ matrix over $\mathbb{F}_q$ is $\Theta(n)$.
Since there is a trivial algorithm that computes the trace by learning the entire matrix using $n$ queries, we focus on the lower bound.

Learning the trace can be regarded as a coset identification problem (defined in \sec{coset}) in the group $G=\mathbb{F}_q^{n\times n}$ with subgroup $H=\{M\in\mathbb{F}_q^{n\times n}: \tr M=0\}\cong \mathbb{F}_q^{n^2-1}$.
The irreducible characters $\chi_Z$ of $H$ are indexed by $Z \in\mathbb{Z}_m^{n\times n}$ with $Z_{nn}=0$, and satisfy $\chi_Z(M) = e(\langle Z,M\rangle)$ where $\langle Z,M\rangle \coloneqq \sum_{i,j=1}^n Z_{ij}M_{ij}$.

\subsubsection{Learning the trace over $\mathbb{F}_2$}

First we consider the case $q=2$.
Then the irreducible characters $\chi_Z$ of $H$ for $Z \in \mathbb{Z}_m^{n \times n}$ (with $Z_{nn}=0$) satisfy
\begin{align}
  \chi_Z(M) = (-1)^{\langle Z, M\rangle}.
\end{align}
For irredicible character $Z$, the induced representation can be decomposed into two irreducible characters of $G$:
\begin{align}\label{eq:reps-M}
  \chi_{Z,0}(M) = (-1)^{\langle Z, M\rangle};\qquad
  \chi_{Z,1}(M) = (-1)^{\langle Z, M\rangle+\tr M}.
\end{align}
It is easy to check that for $M\in G$,
$\chi_{Z,0}(M+E_{nn})=\chi_{Z,0}(M)$
and $\chi_{Z,1}(M+E_{nn})=-\chi_{Z,1}(M)$, where $E_{ij}$ is an $n\times n$ matrix whose entries are zero except that $(E_{ij})_{ij}=1$.
We emphasize that in \eq{reps-M}, $M\in G$ (rather than in $H$ since we are now looking at the representations of the entire group), and $Z_{nn}=0$.

On the other hand, recall that the phase query oracle is $U(M)\colon \ket{x,y}\mapsto (-1)^{y^\top\! Mx}\ket{x,y}$, which is a unitary representation of $M$ with character $\xi(M)\coloneqq\tr(U(M))=\sum_{x,y\in\F_2^n}(-1)^{y^\top\! Mx}$.
To determine the optimal success probability, we calculate the irrep content of $U^{\otimes t}$.
The character of $U^{\otimes t}$ is $\xi^t$, satisfying
\begin{align}\nonumber
  \tr(U^{\otimes t}(M))
  &= \tr(U(M))^t = (\xi(M))^t\\
  &= \left(\sum_{x,y\in\F_2^n}(-1)^{y^\top\! Mx}\right)^t=\sum_{x_1,\ldots,x_t,y_1,\ldots,y_t\in\mathbb{F}_2^{n}}(-1)^{\sum_i y_i M x_i}.
\end{align}
Thus it has non-zero Fourier coefficient at $W$ if and only if $W\in R_t$, where $R_t$ is the set of matrices of rank no more than $t$.

We now check containment of the irreps \eq{reps-M} in $U^{\otimes t}$. We find
\begin{align}\label{eq:z-rt}
  m_{Z,0}^{(t)}=\langle \xi^t, \chi_{Z,0}\rangle >0 &\iff Z\in R_t, &
  m_{Z,1}^{(t)}=\langle \xi^t, \chi_{Z,0}\rangle >0 &\iff Z+\1_n\in R_t.
\end{align}
By \thm{opt-cip}, to succeed with probability better than $1/2$, we must choose a $Z$ such that both $m_{Z,0}^{(t)}>0$ and $m_{Z,1}^{(t)}>0$.
However, now we show this is impossible with $t<n/2$.

\begin{lemma}\label{lem:set-f2}
  The set $\{Z:m_{Z,0}^{(t)}>0\wedge m_{Z,1}^{(t)}>0\}$ is empty for $t<n/2$.
\end{lemma}
\begin{proof}
  We show that the set is non-empty only if $t\geq n/2$.
  Suppose there exists $Z$ such that $m_{Z,0}>0$ and $m_{Z,1}>0$.
  By \eq{z-rt}, $Z\in R_t$ and $Z+\1_n\in R_t$.
  Since the ranks of $Z$ and $Z+\1_n$ are no more than $t$, we conclude that the rank of $\1_n=Z+Z+\1_n$ is no more than $2t$.
  Therefore $t\geq n/2$.
\end{proof}
This implies an $n/2$ lower bound, formally stated as follows.

\begin{lemma}\label{lem:binarytracelb}
  For $t<n/2$, any $t$-query quantum algorithm computing the trace of an $n\times n$ matrix over $\mathbb{F}_2$ succeeds with probability at most $1/2$.
\end{lemma}

\begin{proof}
  By \thm{opt-cip} and \lem{set-f2}, the optimal success probability for a  uniformly random matrix in $\mathbb{F}_2^{n\times n}$ is
  \begin{align}
    \frac{1}{2}\max_Z \sum_{b=0}^1 \delta[m_{Z,b}>0] \leq \frac{1}{2}
  \end{align}
  for $t<n/2$.
\end{proof}

On the upper bound side, we present an $\lceil n/2\rceil$-query quantum algorithm, showing that the above lower bound is achievable.

\begin{lemma}
  In the matrix-vector query model, there exists an $\lceil n/2\rceil$-query quantum algorithm that computes the trace of an $n \times n$ matrix over $\F_2$ with probability 1.
\end{lemma}

\begin{proof}
  First we pad the matrix with one extra zero row and one extra zero column if $n$ is odd, and denote the padded matrix by $M'$.
  Let $\ell=\lceil n/2\rceil$.
  It is clear that one query to $M'\in\mathbb{F}_2^{2\ell\times 2\ell}$ can be simulated using one query to $M$.
  By \thm{opt-cip}, it suffices to find an irreducible character such that both $m_{Z,0}>0$ and $m_{Z,1}>0$.
Now consider
\begin{align}
  Z &= \left[\begin{array}{cc} \mathbbm{1}_{\ell} & 0 \\ 0 & 0\end{array}\right] = \sum_{i=1}^{\ell} e_i e_i^\top,
  &
  Z+\1_{2\ell} &= \left[\begin{array}{cc} 0 & 0 \\ 0 & \1_{\ell}\end{array}\right] = \sum_{i=\ell+1}^{2\ell} e_i e_i^\top.
\end{align}
The algorithm first prepares the state
\begin{align}
  \ket{\psi_0} =
  \frac{1}{\sqrt{2}}\ket{e_1,\ldots,e_{\ell}}\ket{e_1,\ldots,e_{\ell}}
  +\frac{1}{\sqrt{2}}\ket{e_{\ell+1},\ldots,e_{2\ell}}\ket{e_{\ell+1},\ldots,e_{2\ell}}.
\end{align}
Making $\ell$ phase queries in parallel, we have 
\begin{align}\nonumber
  \ket{\psi_M} &= U^{\mvphase}(M')\ket{\psi_0} \\\nonumber
               &=
  \frac{1}{\sqrt{2}}(-1)^{\sum_{i=1}^{\ell}M_{ii}'}\ket{e_1,\ldots,e_{\ell}}\ket{e_1,\ldots,e_{\ell}} \\
  &\qquad +\frac{1}{\sqrt{2}}(-1)^{\sum_{i=\ell+1}^{2\ell} M_{ii}'}\ket{e_{\ell+1},\ldots,e_{2\ell}}\ket{e_{\ell+1},\ldots,e_{2\ell}}.
\end{align}
Measuring in the basis $\{\proj{\psi_0},\proj{\psi_1}\}$, where
\begin{align}
  \ket{\psi_1} =
  \frac{1}{\sqrt{2}}\ket{e_1,\ldots,e_{\ell}}\ket{e_1,\ldots,e_{\ell}}
  -\frac{1}{\sqrt{2}}\ket{e_{\ell+1},\ldots,e_{2\ell}}\ket{e_{\ell+1},\ldots,e_{2\ell}},
\end{align}
the algorithm outputs the trace with probability 1.
\end{proof}

The results of this section are summarized in the following theorem.

\begin{theorem}
  In the matrix-vector query model,
  no quantum algorithm can compute the trace of an $n\times n$ matrix over $\mathbb{F}_2$ with probability better than 1/2 using fewer than $n/2$ queries, and
  there exists a quantum algorithm that succeeds with probability 1 using $\lceil n/2\rceil$ queries.
\end{theorem}

\subsubsection{Learning the trace over $\mathbb{F}_q$}

Now we prove a linear lower bound for the task of learning the trace over $\mathbb{F}_q$.
The proof idea is the same as in the case $q=2$, generalized to any finite field.

\begin{theorem}\label{thm:trace}
  In the matrix-vector query model over $\F_q$,
  computing the trace of an $n \times n$ matrix with probability more than $1/q$ requires at least $n/2$ queries.
\end{theorem}

\begin{proof}
The induced representation of $Z$ (defined in the second paragraph of \sec{trace}) can be decomposed into $q$ 1-dimensional irreps whose characters are
\begin{align}
  \chi_{Z,s}(M) &= e(\langle Z, M\rangle+s\cdot\tr M) = e(\langle Z+s\1_n,M\rangle)
\end{align}
for $s\in\mathbb{F}_q$.
Again, recall that a phase query oracle $U(M)\colon\ket{x,y}\mapsto e(y^\top\! Mx)\ket{x,y}$ is a unitary representation of $M$.
The character of $U$ is the trace $\xi(M) \coloneqq \tr(U(M))=\sum_{x,y\in\Fq^n}e(y^\top\! M x)$.
The optimal success probability is determined by the irrep content of $U^{\otimes t}$, and the character of $U^{\otimes t}$ is $\xi^t$, satisfying
\begin{align}
  \tr(U^{\otimes t}(M))=\xi^t(M) = \sum_{x_1,\ldots,x_t,y_1,\ldots,y_t\in \mathbb{F}_q^n}e\Bigg(\sum_{i=1}^t y_i^\top\! M x_i\Bigg).
\end{align}
Thus for every $s\in\mathbb{Z}_m$,
\begin{align}
  m_{Z,s}^{(t)}=\langle \xi^t,\chi_{Z,s}\rangle > 0 &\iff Z+s\cdot\1_n\in R_t,
\end{align}
where $R_t$ is the set of matrices of rank no more than $t$.
Since $\1_n\notin R_{n-1}$, we conclude for $t<n/2$ the success probability is at most $1/q$, as claimed.
\end{proof}

\subsection{Null space} \label{sec:nullspace}

In this section, we show a linear lower bound on the matrix-vector quantum query complexity of computing the rank of a matrix $M\in\mathbb{F}_q^{m\times n}$ for $m\geq n$.
This is without loss of generality since for $m<n$, by \thm{left-right-eq}, we can simulate oracle access to $M^\top$ using one query to $M$.

The rank problem is an instance of the hidden subgroup problem (HSP) over $\F_q^m$ since two vectors map to the same value if and only if their difference is in the null space.
However, the lower bound for the abelian HSP \cite{KNP07} does not directly apply to this problem since the instance is more structured---specifically, the subgroup hiding function is a linear transformation.

We recall some standard facts from linear algebra over finite fields.
For $\ell\geq m$, let $\binom{\ell}{m}_q \coloneqq \frac{\prod_{i=0}^{m-1} (q^\ell-q^i)}{\prod_{i=0}^{m-1}(q^m-q^i)}$ denote a Gaussian binomial coefficient.

\begin{lemma}\label{lem:numsubspaces}
  The number of $m$-dimensional subspaces of an $\ell$-dimensional space over $\F_q$ is $\binom{\ell}{m}_q$.
\end{lemma}

\begin{lemma}\label{lem:numsubspacescontainingv}
  For integers $k\leq m\leq\ell$ and any $k$-dimensional space $V$ over $\F_q$, the number of $m$-dimensional subspaces of an $\ell$-dimensional space containing $V$ is $\binom{\ell-k}{m-k}_q$.
\end{lemma}

For proofs of these facts, see for example \cite[Lemma 9.3.2]{BCN89}.

\paragraph{Computing the rank.}
Now we consider the problem of computing the rank of a matrix $M\in\mathbb{F}_q^{m\times n}$ for $m\geq n$.
A matrix $M$ has rank $r$ if and only if its null space is $(n-r)$-dimensional.

By \lem{polynomial-2t}, the success probability of a $t$-query algorithm is a degree-$2t$ polynomial in $\delta_{xy}$.
This polynomial $P$ can be written as
\begin{align}
  P(\delta) = \sum_{S\subseteq \mathbb{F}_q^n\times\mathbb{F}_q^m} c_S \prod_{(x,y)\in S} \delta_{xy},
\end{align}
with $c_S=0$ for $|S|>\deg(P)$.
For an input $M$, the assignments to these variables are $\delta_{xy}=\delta[Mx=y]$; we will sometimes write $\delta_{xy}=\delta_{xy}(M)$ to emphasize that $\delta$ is a function of $M$.

Now symmetrize by averaging over all matrices with nullity $d$, giving
\begin{align}\nonumber
  Q(d) &\coloneqq \Exp_{M\sim Y_d} [P(\delta(M))] \\\nonumber
  &= \sum_{S\subseteq\mathbb{F}_q^n\times\mathbb{F}_q^m} c_S \Exp_{M\sim Y_d}\Bigg[\prod_{(x,y)\in S}\delta_{xy}(M)\Bigg] \\\label{eq:sum-probability}
  &= \sum_{S\subseteq\mathbb{F}_q^n\times\mathbb{F}_q^m} c_S \Pr_{M\sim Y_d}[Mx=y~\forall (x,y)\in S],
\end{align}
where $Y_d$ is the set of matrices of nullity $d$.
Here $M$ is drawn uniformly from $Y_d$.
Since $0\leq P(\delta(M))\leq 1$, we have
$0\leq Q(d)\leq 1$.
The following lemma states that we can approximate $Q(d)$ with a low-degree polynomial.
Van Apeldoorn and Gribling previously showed the same statement in their proof of a lower bound for Simon's problem for linear functions \cite[Lemma~3]{vAG18}.
That problem can be viewed as a special case of our problem with $m=n$. We observe that essentially the same proof establishes this lemma for $m\geq n$.
\begin{lemma}\label{lem:symmetrization}
  There exists a polynomial $R$ of degree at most $2t$ such that for each $d\in[n]$, $R(q^d)=Q(d)$.
\end{lemma}

We emphasize that we do not bound the degree of $Q(d)$ because we do not know how to represent it as a polynomial in $d$.
Instead, the lower bound is established by showing (i) a lower bound on the degree of the polynomial $R$ and (ii) that the degree of $R$ is no more than $2t$.

Next, recall a lemma by Koiran, Nesme, and Portier \cite[Lemma~5]{KNP07}.

\begin{lemma}\label{lem:knp}
  Let $c>0$ and $\xi>1$ be constants and let $f$ be a real polynomial with the following properties:
  \begin{enumerate}[nosep]
  \item for any integer $0\leq \xi\leq n$, $|f(\xi^i)|\leq 1$;
  \item for some real number $1\leq x_0\leq\xi$, $|f'(x_0)|\geq c$.
  \end{enumerate}
  Then $\deg f=\Omega(n)$.
\end{lemma}

\lem{symmetrization} and \lem{knp} imply an $\Omega(\min\{m,n\})$ lower bound for distinguishing a matrix is full-rank or has nullity 1.
The case $m=n$ was previously shown by van Apeldoorn and Gribling \cite[Theorem~1]{vAG18}.
We briefly explain the main ideas for completeness.
By \lem{symmetrization}, for $d\in\{0,1,\ldots,n-1\}$, $R(q^d)=Q(d)$ and $\deg(R)\leq 2t$.
For distinguishing a full-rank matrix (i.e., $d=0$) from a rank $n-1$ matrix (i.e, $d=1$), we set $R(1)\geq 1-\epsilon$ and $R(q)\leq\epsilon$.
There exists $x_0\in[1,q]$ such that $R'(x_0)\geq \frac{|R(q)-R(1)|}{q-1}\geq\frac{1-2\epsilon}{q-1}$.
By \lem{knp}, $t=\Omega(n)$ for $m\geq n$.
For $m<n$, an $\Omega(m)$ lower bound follows from \thm{left-right-eq}.
Overall, this gives the following.

\begin{theorem}\label{thm:rank-decision}
  The bounded-error matrix-vector quantum query complexity of deciding if an $m \times n$ matrix over $\F_q$ is full-rank is $\Omega(\min\{m,n\})$.
    In particular, $\Omega(\min\{m,n\})$ queries are needed to decide whether the matrix is full-rank or has nullity $1$.
\end{theorem}

There is a trivial algorithm that learns an entire $m \times n$ matrix using $\min\{m,n\}$ queries. Thus the query complexity of computing the rank is $\Theta(\min\{m,n\})$.

\begin{corollary}\label{cor:rank}
  The bounded-error query matrix-vector quantum complexity of computing the rank of an $m\times n$ matrix over $\Fq$ is $\Theta(\min\{m,n\})$.
\end{corollary}

With the same argument, the quantum query complexity of computing the determinant of an $n\times n$ matrix over $\Fq$ is $\Theta(n)$.
Moreover, the classical query complexity is $\Theta(n^2)$, implied by the $\Omega(n^2)$ lower bound for rank testing by Rashtchian, Woodruff, and Zhu \cite[Theorem~3.3]{RWZ20}.
\begin{corollary}[Determinant]\label{cor:determinant}
  The bounded-error classical and quantum query complexities of computing the determinant of an $n\times n$ matrix over $\mathbb{F}_q$ through matrix-vector products are $\Theta(n^2)$ and $\Theta(n)$, respectively.
\end{corollary}

\subsection{Solving linear systems} \label{sec:regression}

In this section, we consider the quantum query complexity of solving the linear system $Ax=b$ for $A\in\Fq^{n\times n}$ is $\Theta(n)$.
Since there is an $n$-query algorithm learning the entire matrix using $n$ matrix-vector queries, we focus on the lower bound.

Our proof is based on a randomized reduction from deciding whether a submatrix is full rank.
For a square matrix $A$, let $A^{ij}$ be the submatrix obtained by deleting the $i^{th}$ row and the $j^{th}$ column, and let $A_{ij}$ denote the $(i,j)$ element of $A$.
The elements of $A^{-1}$ can be computed as
\begin{align}
  (A^{-1})_{ij} = \frac{\det A^{ij}}{\det A}.
\end{align}
Given an invertible $A$, one can use a linear system solver to decide whether $(A^{-1})_{11}$ is non-zero, and thus decide if the minor $A^{11}$ is full-rank.

In our reduction, to decide whether $M\in\Fq^{n\times n}$ is full-rank given access to matrix-vector products, we pad $M$ with one extra random row and one extra random column, giving a matrix $A \in \F_q^{(n+1)\times(n+1)}$.
We show that with sufficiently high probability, the padded matrix is full-rank.
Thus, invoking a linear system solver with $b=e_1$, we learn whether $\det M=0$.
Thus the linear regression lower bound follows from \thm{rank-decision}.

\begin{theorem}\label{thm:regression}
  The bounded-error matrix-vector quantum query complexity of solving an $n \times n$ linear system is $\Omega(n)$.
\end{theorem}

\begin{proof}
  Assume toward contradiction that $\A$ is a $t$-query quantum algorithm for determining whether $(A^{-1})_{11}$ is non-zero for any invertible $A\in\mathbb{F}_q^{(n+1)\times(n+1)}$, succeeding with probability $p \ge 1/3$ with $t=o(n)$.
  We present a $t$-query algorithm for determining whether an $n\times n$ matrix is full-rank with probability $p(1-1/q)^{2}\geq 1/12$.

  Given access to $M\in\Fq^{n\times n}$, the algorithm first samples two random vectors $u,v\in\Fq^n$ and a random element $a\in\Fq$ to give the padded matrix
  \begin{align}\label{eq:matrix-A}
    A = \left[
    \begin{array}{cc}
      a & u^\top \\ v & M
    \end{array}
    \right].
  \end{align}
  The matrix-vector product $A(x_0,x^\top)^\top$ for $x_0\in\Fq,x\in\Fq^n$ can be computed using one $\mv$ query to $Mx$ since
  \begin{align}
    A \left[\begin{array}{cc} x_0\\ x\end{array}\right] =
    \left[\begin{array}{cc} a_0 + u^\top\! x \\ x_0 v + Mx \end{array}\right].
  \end{align}

  We show that with probability at least $(1-1/q)^2$, the matrix $A$ is invertible (i.e., $\det A\neq 0$) given that $\mathrm{rank}(M)\geq n-1$.
  If $M$ is invertible, the submatrix $B=(v, M)$ is full-rank.
  If $\mathrm{rank}(M)=n-1$, then without loss of generality, we consider the case that the first $n-1$ rows of $M$ are linearly independent, and the last row is a linear combination of the first $n-1$ rows,
  since other cases can be handled accordingly by rearranging the rows.
  We let
  \begin{align}
    M = \left[
    \begin{array}{c}
      M' \\ w^\top
    \end{array}
    \right].
  \end{align}
  for an $(n-1)\times n$ matrix $M'$ and an $n\times 1$ vector $w$.
  Since $w^\top$ is a linear combination of the first $n-1$ rows, we write $w^\top=c^\top\! M'$ for an $(n-1)\times 1$ vector $c$.
  Since $M'$ is full-rank, the vector $c$ satisfying $w^\top=c^\top\! M'$ is unique. Now write the vector
  \begin{align}
    v = \left[
    \begin{array}{c}
      z \\ b
    \end{array}
    \right]
  \end{align}
  for an $(n-1)\times 1$ matrix $z$ and $b\in\mathbb{F}_q$.
  The matrix $B$ is not full rank if and only if the last row is a linear combination of the first $n-1$ rows, i.e., $c^\top\! z=b$, since the first $n-1$ rows of $B$ are linearly independent.
  Since $v$ is a random vector with each element chosen independently, we have
  \begin{align}
    \Pr[\text{$B$ is not full-rank}] = \Pr_{z,b}[c^\top\! z=b] = 1/q.
  \end{align}
  Thus with probability at least $1-1/q$ the matrix $B$ is full-rank.

  Conditioned on $B$ being full-rank, the matrix $A$ is not full-rank if and only if the vector $(a,u^\top)$ is in the vector space spanned by the rows of $B$.
  The number of vectors in the vector space is $q^{(n-1)}$.
  Thus
  \begin{align}
    \Pr_{a,u,v}[\text{$A$ is not full-rank} \mid \text{$B$ is full-rank}] = 1/q.
  \end{align}
  Therefore with probability at least $1-1/q$, $A$ is invertible.
  Conditioned on successfully simulating $\mv$ queries of an invertible $A$, the algorithm $\A$ determines whether $(A^{-1})_{11}$ is nonzero with probability $p$.
  Thus the algorithm succeeds with probability at least $p(1-1/q)^2\geq 1/12$ using $t=o(n)$ queries to $M$.
  By \thm{rank-decision} we have a contradiction.
\end{proof}

The same proof idea shows that a lower bound for rank testing implies a lower bound for linear regression in the $\vmv$ model.
Rashtchian, Woodruff, and Zhu show that the query complexity of distinguishing rank-$n$ matrices from rank-$(n-1)$ matrices over $\Fq$ is $\Omega(n^2)$ \cite[Theorem~3.3]{RWZ20}.

\begin{corollary}\label{cor:regression-vmv}
  The bounded-error classical $\vmv$ query complexity of solving an $n\times n$ linear system over $\Fq$ is $\Omega(n^2)$.
\end{corollary}

\begin{proof}
  By the same idea as in the proof of \thm{regression}, it suffices to show that one $\vmv$ query to the $(n+1)\times(n+1)$ matrix $A$ in \eq{matrix-A} can be simulated with one $\vmv$ query to the $n\times n$ matrix $M$.
  For any query $x,y$, we let $x=(x_0,x_1^\top)^\top$ and $y=(y_0,y_1^\top)^\top$
  for $n\times 1$ matrices $x_1,y_1$.
  The product $y^\top\! A x$ can be computed using one $\vmv$ query to $M$ since $y^\top\! A x = ay_0x_0 + y_0 u^\top\! x_1 + y_1^\top\! v x_0 + y_1^\top\! M x_1$.
  Since no $o(n^2)$-query classical algorithm can distinguish rank-$n$ matrices from rank-$(n-1)$ matrices \cite[Theorem~3.3]{RWZ20}, the bounded-error query complexity of solving linear systems is $\Omega(n^2)$.
\end{proof}

\subsection{Rank testing} \label{sec:rank}

In this section, we show a linear lower bound on distinguishing whether an $m\times n$ matrix $M$ has $\mathrm{rank}(M)=n$ or $\mathrm{rank}(M)\leq n/2$, where $m\geq n$.
First we show the following lemma using ideas from \cite{KNP07}.

\begin{lemma}\label{lem:rank-testing-polynomial-degree}
  Let $\xi\geq 2$ and let $n$ be an even integer. Then any polynomial $f$ satisfying
  \begin{enumerate}[nosep]
  \item $0\leq f(\xi^i)\leq 1$ for $i\in\{0,1,\ldots,n-1\}$ and
  \item $f(1)\leq 1/3$ and $f(\xi^{i})\geq 2/3$ for $i\in\{n/2,n/2+1,\ldots,n-1\}$
  \end{enumerate}
  has $\deg(f)=\Omega(n)$.
\end{lemma}

\begin{proof}
  Let $d=\deg(f)$.
  Toward contradiction, we assume $d=o(n)$.
  For intervals $S_i \coloneqq [\xi^i,\xi^{i+1})$, 
  since $\deg(f'),\deg(f'')=o(n)$, there exists an index $a\in\{9n/10,\ldots,n-3,n-2\}$ such that none of the roots of $f'$ and $f''$ has its real part in $S_a$.
  This implies that $f'$ is monotonically increasing or decreasing in $S_a$, i.e., $f$ is concave or convex.
  In each case, $f(\frac{\xi^a+\xi^{a+1}}{2})\in[0,1]$.
  If $f$ is convex in $S_a$,
  \begin{align}
   \left| f'\Bigl(\frac{\xi^a+\xi^{a+1}}{2}\Bigr)\right| \leq \frac{1}{\xi^{a+1}-\frac{\xi^{a+1}+\xi^a}{2}} = \frac{2}{\xi^{a+1}-\xi^{a}} \leq \frac{2}{\xi^a}\leq 2\xi^{-9n/10}.
  \end{align}
  If $f$ is concave in $S_a$, reflecting about the $x$-axis gives the same bound.

  By the second constraint, there exists $x_0\in[1,\xi^{n/2}]$ such that
  \begin{align}
    |f'(x_0)| \geq \frac{|f(\xi^{n/2})-f(1)|}{\xi^{n/2}-1} \geq \xi^{-n/2}/3.
  \end{align}
  Therefore
  \begin{align}\label{eq:upper-fp}
    \left|\frac{f'(\frac{\xi^a+\xi^{a+1}}{2})}{f'(x_0)}\right| \leq 6\xi^{-2n/5} \leq \xi^{3-2n/5}.
  \end{align}

  On the other hand, since $\deg(f')=d-1$, denoting the roots $a_1,\ldots,a_{d-1}\in\mathbb{C}$, we write
  \begin{align}
    f'(x) = \lambda\prod_{i=1}^{d-1} (x-a_i).
  \end{align}
  Thus
  \begin{align}
    \left|\frac{f'(\frac{\xi^a+\xi^{a+1}}{2})}{f'(x_0)}\right|
    = \prod_{i=1}^{d-1} \Biggl|\frac{\frac{\xi^a+\xi^{a+1}}{2} - a_i}{x_0-a_i} \Biggr| = \prod_{i=1}^{d-1} |g(a_i)|,
  \end{align}
  where
  \begin{align}
    g(x) = \frac{x-\frac{\xi^a+\xi^{a+1}}{2}}{x-x_0}.
  \end{align}

  Our goal is to show that for each $i$, $|g(a_i)|\geq \frac{1}{2\xi}$.
  Recall that for each $i$, $\Re(a_i)\notin S_a$.
  Also for real $x\notin S_a$, $x\geq x_0$, we have $|g(x)|\geq \frac{\xi-1}{2\xi}\geq \frac{1}{2\xi}$.
  For real roots, $|g(a_i)|\geq \frac{1}{2\xi}$.
  Now we consider the case where $a_i=\alpha+\beta i$ for $\beta\neq 0$, giving
  \begin{align}
    |g(\alpha+\beta i)|^2 = \frac{(\alpha-\frac{\xi^a+\xi^{a+1}}{2})^2+\beta^2}{(\alpha-x_0)^2+\beta^2}.
  \end{align}
  If $(\alpha-\frac{\xi^a+\xi^{a+1}}{2})^2\geq (\alpha-x_0)^2$, then $|g(\alpha+\beta i)|\geq 1$.
  Otherwise,
  \begin{align}
    |g(\alpha+\beta i)| \geq \Bigg|\frac{\alpha-\frac{\xi^a+\xi^{a+1}}{2}}{\alpha-x_0}\Bigg| \geq \frac{1}{2\xi}.
  \end{align}
  We have shown that $|g(a_i)|\geq \frac{1}{2\xi}$ for every root $a_i$.
  Now we have
  \begin{align}
    \Bigg|\frac{f'(\frac{\xi^a+\xi^{a+1}}{2})}{f'(x_0)}\Bigg| =\prod_{i=1}^{d-1} |g(a_i)| \geq (2\xi)^{-d+1} \geq \xi^{2-2d}.
  \end{align}
  Thus by \eq{upper-fp}, we have $\xi^{3-2n/5}\geq\xi^{2-2d}$ and conclude $d\geq n/5-1/2=\Omega(n)$---a contradiction.
\end{proof}

\lem{symmetrization} and \lem{rank-testing-polynomial-degree} imply the following theorem.

\begin{theorem}\label{thm:rank-testing}
  The bounded-error matrix-vector quantum query complexity of determining whether a matrix $M\in\mathbb{F}_q^{m\times n}$ has $\mathrm{rank}(M)=n$ or $\mathrm{rank}(M)\leq n/2$ is $\Omega(n)$.
\end{theorem}

\section*{Acknowledgments}

We thank Robin Kothari for bringing our attention to work on classical algorithms in the matrix-vector and vector-matrix-vector query models, and for providing feedback on an initial version of this paper. We thank Max Simchowitz and Blake Woodworth for a discussion that clarified aspects of their paper \cite{BHSW20}, and Jialin Zhang for clarifications of her paper \cite{SWYZ19}. We also thank Ashley Montanaro for pointing out connections to his paper \cite{MS20}, and Joran van Apeldoorn and Sander Gribling for a discussion of their paper \cite{vAG18}.

AMC and SHH acknowledge support from the Army Research Office (grant W911NF-20-1-0015); the Department of Energy, Office of Science, Office of Advanced Scientific Computing Research, Quantum Algorithms Teams and Accelerated Research in Quantum Computing programs; and the National Science Foundation (grant CCF-1813814). TL acknowledges support from ARO contract W911NF-17-1-0433, NSF grant PHY-1818914, and a Samsung Advanced Institute of Technology Global Research Partnership.

\end{document}